\documentclass[letterpaper, 10 pt, conference]{ieeeconf} 

\IEEEoverridecommandlockouts                              
\title{\LARGE \bf
Strong Coordination with Polar Codes
}

\author{Matthieu R. Bloch$^{1}$, Laura Luzzi$^{2}$, and J{\"o}rg Kliewer$^{3}$
\thanks{$^{1}$M. Bloch is with the School of Electrical and Computer Engineering, Georgia Institute of Technology, Atlanta, GA 30332, USA and with the GT-CNRS UMI 2958, Metz, France.
        {\tt\small matthieu.bloch@ece.gatech.edu}}%
\thanks{$^{2}$L. Luzzi was with the Department of Electrical and Electronic Engineering,
Imperial College London, London SW7 2AZ, United Kingdom.
She is now with Laboratoire ETIS (ENSEA - Universit{\'e} de Cergy-Pontoise
- CNRS), 6 Avenue du Ponceau, 95014 Cergy-Pontoise, France. {\tt\small laura.luzzi@ensea.fr}}%
\thanks{$^{3}$J. Kliewer is with the Klipsch School of Electrical and Computer Engineering, New Mexico State University, Las Cruces, NM, 88003, USA. {\tt\small  jkliewer@nmsu.edu}}%
}

\usepackage{color,graphicx,amsmath,amssymb,stmaryrd,dsfont}
\usepackage{rotating,array,url}

\DeclareMathAlphabet{\eurm}{U}{eur}{m}{n}
\DeclareMathAlphabet{\mathbsf}{OT1}{cmss}{bx}{n}
\DeclareMathAlphabet{\mathssf}{OT1}{cmss}{m}{sl}
\DeclareMathAlphabet{\mathcsf}{OT1}{cmss}{sbc}{n}

\newcommand{\randomvalue}[1]{\eurm{\uppercase{#1}}}

\DeclareSymbolFont{bsfletters}{OT1}{cmss}{bx}{n}  
\DeclareSymbolFont{ssfletters}{OT1}{cmss}{m}{n}
\DeclareMathSymbol{\bsfGamma}{0}{bsfletters}{'000}
\DeclareMathSymbol{\ssfGamma}{0}{ssfletters}{'000}
\DeclareMathSymbol{\bsfDelta}{0}{bsfletters}{'001}
\DeclareMathSymbol{\ssfDelta}{0}{ssfletters}{'001}
\DeclareMathSymbol{\bsfTheta}{0}{bsfletters}{'002}
\DeclareMathSymbol{\ssfTheta}{0}{ssfletters}{'002}
\DeclareMathSymbol{\bsfLambda}{0}{bsfletters}{'003}
\DeclareMathSymbol{\ssfLambda}{0}{ssfletters}{'003}
\DeclareMathSymbol{\bsfXi}{0}{bsfletters}{'004}
\DeclareMathSymbol{\ssfXi}{0}{ssfletters}{'004}
\DeclareMathSymbol{\bsfPi}{0}{bsfletters}{'005}
\DeclareMathSymbol{\ssfPi}{0}{ssfletters}{'005}
\DeclareMathSymbol{\bsfSigma}{0}{bsfletters}{'006}
\DeclareMathSymbol{\ssfSigma}{0}{ssfletters}{'006}
\DeclareMathSymbol{\bsfUpsilon}{0}{bsfletters}{'007}
\DeclareMathSymbol{\ssfUpsilon}{0}{ssfletters}{'007}
\DeclareMathSymbol{\bsfPhi}{0}{bsfletters}{'010}
\DeclareMathSymbol{\ssfPhi}{0}{ssfletters}{'010}
\DeclareMathSymbol{\bsfPsi}{0}{bsfletters}{'011}
\DeclareMathSymbol{\ssfPsi}{0}{ssfletters}{'011}
\DeclareMathSymbol{\bsfOmega}{0}{bsfletters}{'012}
\DeclareMathSymbol{\ssfOmega}{0}{ssfletters}{'012}

\newcommand{\rvS}{{\randomvalue{S}}}	
\newcommand{\rvU}{{\randomvalue{U}}}	
\newcommand{\rvV}{{\randomvalue{V}}}	
\newcommand{\rvX}{{\randomvalue{X}}}  	
\newcommand{\rvY}{{\randomvalue{Y}}}	

\newcommand{\calB}{{\mathcal{B}}}
\newcommand{\calC}{{\mathcal{C}}}

\newcommand{\calG}{{\mathcal{G}}}

\newcommand{\calU}{{\mathcal{U}}}

\newcommand{\calX}{{\mathcal{X}}}
\newcommand{\calY}{{\mathcal{Y}}}

\newcommand{\E}[2][]{{\mathbb{E}_{#1}}{\left(#2\right)}}       
       
\newcommand{\avgD}[2]{{{\mathbb{D}}\!\left({#1\Vert#2}\right)}}
\newcommand{\V}[1]{{{\mathbb{V}}\!\left(#1\right)}}
\newcommand{\avgI}[1]{{{\mathbb{I}}\!\left(#1\right)}}

\newcommand{\Hb}[1]{{\mathbb{H}_b}\left(#1\right)}

\newcommand{\card}[1]{\ensuremath{\left|{#1}\right|}}           
\newcommand{\abs}[1]{\ensuremath{\left|#1\right|}}              
\newcommand{\eqdef}{\ensuremath{\triangleq}}                    
\newcommand{\intseq}[2]{\ensuremath{\llbracket{#1},{#2}\rrbracket}}  

\newtheorem{lemma}{Lemma}

\newtheorem{theorem}{Theorem}
\newtheorem{proposition}{Proposition}

\newtheorem{remark}{Remark}

\renewcommand{\leq}{\leqslant}
\renewcommand{\geq}{\geqslant}
\newcommand{\db}[1]{{\Bar{#1}}}
\graphicspath{{graphics/}}

\begin{document}

\maketitle

\begin{abstract}
In this paper, we design explicit codes for strong coordination in two-node networks. Specifically, we consider a two-node network in which the action imposed by nature is binary and uniform, and the action to coordinate is obtained via a symmetric discrete memoryless channel. By observing that polar codes are useful for channel resolvability over binary symmetric channels, we prove that nested polar codes achieve a subset of the strong coordination capacity region, and therefore provide a constructive and low complexity solution for strong coordination.
\end{abstract}

\section{Introduction}
\label{sec:introduction}

The characterization of the information-theoretic limits of coordination in networks has recently been investigated in~\cite{Cuff2010}. The coordinated actions of nodes in the network are modeled by joint probability distributions, and the level of coordination is measured in terms of how well these joint distributions approximate a target joint distribution. Two types of coordination have been introduced: {empirical} coordination, which only requires the empirical distribution of coordinated actions to approach a target distribution, and strong coordination, which requires the total variational distance of coordinated actions to approach a target distribution. The concept of coordination sheds light into the fundamental limits of several problems, such as distributed control or task assignment in a network.

The design of practical and efficient coordination schemes approaching the fundamental limits predicted by information theory has attracted little attention to date. One of the hurdles faced for code design is that the metric to optimize is not a probability of error but a variational distance between distributions. Nevertheless, polar codes~\cite{Arikan2009} have recently been successfully adapted~\cite{Blasco-Serrano2012} for empirical coordination, with an analysis based on results from lossy source coding with polar codes~\cite{Korada2010}. In this paper, we construct polar codes that are able to achieve strong coordination in some cases. Unlike the construction in~\cite{Blasco-Serrano2012}, which solely relies on source coding with polar codes, our construction also exploits polar codes for channel resolvability~\cite{Han1993}. Channel resolvability characterizes the bit rate required to simulate a process at the output of a channel and plays a key role in the analysis of the common information between random variables~\cite{Wyner1975a}, secure communication over wiretap channels~\cite{Hayashi2006,Bloch2011a}, and coordination~\cite[Lemma 19]{Cuff2010}. By remarking that polar codes can be used for channel resolvability, we are able to provide a constructive alternative to the information-theoretic proof in~\cite{Cuff2010}.

The remainder of the paper is organized as follows. Section~\ref{sec:preliminaries} sets the notation and recalls known results for polar codes. Section~\ref{sec:polar-chann-resolv} shows that polar codes provide channel resolvability for symmetric channels by leveraging results in~\cite{Mahdavifar2011}. Section~\ref{sec:polar-strong-coord} proves that polar codes achieve strong coordination for simple two-node networks with symmetric actions. Finally, Section~\ref{sec:discussion} concludes the paper with a discussion of potential improvements and extensions.

\section{Notation and Preliminaries}
\label{sec:preliminaries}

First, a word about notation. Given a length $n$ vector $\mathbf{x}=(x_1,\cdots,x_n)$ and $i\in\intseq{1}{n}$, we use the notation $\mathbf{x}_1^{i}$ as a shorthand for the row vector $(x_1,\cdots,x_i)$. Similarly, for any set $F\in\intseq{1}{n}$, we denote by $\mathbf{x}_F$ the vector of length $\abs{F}$ containing the indices $x_i$ for $i\in F$. The distributions of different random variables defined on the same alphabet $\calX$ are denoted by different symbols, e.g. $p_{\rvX}$, $q_\rvX$. For brevity, the subscripts in the distributions may be dropped if the alphabet is clear from the context or from the argument. We also use $\mathbb{D}$ to denote the Kullback-Leibler divergence between two distributions.

Next, we briefly review the concepts and notation related to polar codes that will be used throughout the paper. The key element in the polar coding construction is the decomposition of $n\eqdef 2^m$ independent copies of a given binary-input discrete-memoryless channel $(\calX,W_{\rvY|\rvX},\calY)$ with capacity $C(W_{\rvY|\rvX})$ into $n$ \emph{bit-channels} which are essentially either error-free or pure noise channels. Specifically, consider the transformation $G_n\eqdef G_2^{\otimes n}$ where
\begin{align*}
  G_2\eqdef \begin{pmatrix} 1 & 0 \\ 1 & 1 \end{pmatrix}
\end{align*}
and $\otimes$ is the Kronecker product. A vector $\mathbf{u}\in\{0,1\}^n$ is transformed into $\mathbf{x} = \mathbf{u}G_n$.
The $i$-th bit channel $(\{0,1\},W_n^{(i)},\calY^n\times\{0,1\}^{i-1})$ is a composite channel that combines the transformation $G_n$ and the channel, and is defined by its transition probabilities
\begin{align*}
  W_n^{(i)}(\mathbf{y},\mathbf{u}_{1}^{i-1}|u_i)\eqdef \frac{1}{2^{n-1}}\sum_{\mathbf{u}_{i+1}^n}W_{\rvY^n|\rvX^n}(\mathbf{y}|\mathbf{u}G_n).
\end{align*}
For $n$ large enough, the bit channels \emph{polarize}, i.e. they become either completely noisy or noise-free. The exact measure of the noise level will be specified in subsequent sections.

\section{Channel Resolvability with Polar Codes}
\label{sec:polar-chann-resolv}

\subsection{Channel resolvability}
\label{sec:chann-resolv}

In its simplest formulation, the problem of \emph{channel resolvability}~\cite{Wyner1975a,Han1993} can be stated as follows. Consider a discrete memoryless channel $(\calX,W_{\rvY|\rvX},\calY)$ whose input is an i.i.d. source distributed according to $q_\rvX$; the output of the channel is then an i.i.d. process distributed according to $q_\rvY$. The aim is to construct a sequence of codes $\{\mathcal{C}_n\}_{n\geq 1}$ of rate $R$ and increasing block length $n$, such that the output distribution $p_{\rvY^n}$ induced by a uniform choice of the codewords in $\calC_n$ approaches the distribution $q_{\rvY^n}\sim\prod_{i=1}^n q_\rvY$ in variational distance, i.e.
\begin{equation} \label{variational_distance}
\lim_{n\to \infty} \V{p_{\rvY^n},q_{{\rvY}^n}} = 0.
\end{equation}
In this case, the sequence $\{\mathcal{C}_n\}_{n\geq 1}$ is called a sequence of \emph{resolvability codes} achieving \emph{resolution rate} $R$ for $(W_{\rvY|\rvX},q_{\rvX})$. The \emph{channel resolvability} of $W_{\rvY|\rvX}$ is then defined as the minimum resolution rate such that resolvability codes exist for any input source.

\subsection{Coding scheme for channel resolvability}
\label{sec:coding-scheme-chann}

In this section, we leverage the results of~\cite{Mahdavifar2011} to construct resolvability codes when $(\calX,W_{\rvY|\rvX},\calY)$ is a binary-input symmetric DMC and $q_\rvX$ is the uniform distribution on $\{0,1\}$, i.e. $q_\rvX\sim\calB(\frac{1}{2})$; this result will be exploited in Section~\ref{sec:polar-strong-coord} for the problem of coordination. We use the notion of symmetry in~\cite{InformationTheoryReliableCommunication}, according to which there exists a permutation $\pi_1:\mathcal{Y} \to \mathcal{Y}$ such that $\pi_1=\pi_1^{-1}$ and 
\begin{equation} \label{symmetry}
\forall y \in \mathcal{Y}, \quad W_{\rvY|\rvX}(y|0)=W_{\rvY|\rvX}(\pi_1(y)|1)
\end{equation} 
In particular, the following property of symmetric channels will be useful. 

\begin{lemma}[\cite{InformationTheoryReliableCommunication}]
\label{lm:symmetric_channels}
  If $(\calX,W_{\rvY|\rvX},\calY)$ is a memoryless symmetric channel and if $q_{{\rvY}}$ is the output distribution corresponding to the uniform input distribution $q_\rvX$ on $\mathcal{X}$, then
  \begin{align*}
    \forall x \in \mathcal{X}, \quad
    C(W_{\rvY|\rvX})=\avgD{W_{\rvY|\rvX=x}}{q_{{\rvY}}},
  \end{align*}
  where $W_{\rvY|\rvX=x}$ is the output distribution induced by the fixed symbol $x$.
\end{lemma}

Let $W_n^{(i)}$ denote the set of bit channels corresponding to $W_{\rvY|\rvX}$, and define the sets of ``good bits'' $\mathcal{G}_n$ and ``bad bits'' $\mathcal{B}_n$ as
\begin{align*}
  \calG_n &\eqdef\left\{i\in\intseq{1}{n}:C(W_n^{(i)})\geq 2^{-n^\beta}\right\},\\
  \text{and }\calB_n&\eqdef\intseq{1}{n}\setminus \calG_{n}.
\end{align*}

Our strategy to simulate the i.i.d. process distributed according to $q_{\rvY}$ is to send random uniform bits on the good bits, and fixed bits on the bad bits. Intuitively, the uniform bits will be preserved by the noiseless bit-channels, while the pure noise bit-channels will produce almost-uniform bits for any input. Formally, let $r=\abs{\mathcal{G}_n}$ and consider the polar codes defined in Section~\ref{sec:preliminaries}. We will use the $(n,r,\mathcal{G}_n,\mathbf{0}^{n-r})$ \emph{coset code} $\mathcal{C}_n$~\cite{Arikan2009} obtained by using $\mathcal{G}_n$ as the set of information bits and $\mathcal{B}_n$ as the set of \emph{frozen} bits. \smallskip{}
 
\begin{proposition}
\label{prop:resolvablity_polar}
If the channel $(\calX,W_{\rvY|\rvX},\calY)$ is symmetric and $q_{\rvX}\sim\calB(\frac{1}{2})$, then $\{\mathcal{C}_n\}_{n\geq 1}$ is a sequence of resolvability codes of resolution rate $C(W_{\rvY|\rvX})$ for $(W_{\rvY|\rvX},q_{\rvX})$.
\end{proposition}\smallskip{}

\begin{proof}

We know from~\cite[Proposition 20]{Mahdavifar2011} that
\begin{align*}
  \lim_{n\rightarrow\infty }\frac{r}{n} = C(W_{\rvY|\rvX}),
\end{align*}
so that the condition regarding the resolution rate is satisfied. Following~\cite{Mahdavifar2011}, given two vectors $\mathbf{x}^r \in \{0,1\}^r$ and $\mathbf{s}^{n-r} \in \{0,1\}^{n-r}$, we let $(\mathbf{x}^r,\mathbf{s}^{n-r})$ denote the vector $\mathbf{v}^n \in \{0,1\}^n$ such that $\mathbf{v}_{|\mathcal{G}_n}=\mathbf{x}^r$ and $v_{|\mathcal{B}_n}=\mathbf{s}^{n-r}$. We then define a composite channel $(\{0,1\}^{n-r},W_{\rvY^n|\rvS^{n-r}},\mathcal{Y}^n)$, which includes the polar code and the random bits sent on the good bits $\mathcal{G}_n$, so that
\begin{multline*}
  W_{\rvY^n|\rvS^{n-r}}(\mathbf{y}^n|\mathbf{s}^{n-r})\\
  \eqdef\frac{1}{2^r} \sum_{\mathbf{x}^r \in \{0,1\}^r}
  W_{\rvY^n|\rvX^n}\left(\mathbf{y}^n\Bigg|(\mathbf{x}^r,\mathbf{s}^{n-r})G_n\right).
\end{multline*}
It is shown in~\cite[Proposition 13]{Mahdavifar2011} that $W_{\rvY^n|\rvS^{n-r}}$ is symmetric and that
\begin{align*}
  C(W_{\rvY^n|\rvS^{n-r}})\leq \sum_{i \in \mathcal{B}_n} C(W_n^{(i)})\leq(n-r)2^{-n^{\beta}}.
\end{align*}
We now show that this last inequality implies that $\{\calC_n\}_{n\geq 1}$ form a sequence of resolvability codes.

By the definition of coset codes~\cite{Arikan2009}, the output distribution $p_{\rvY^n}$ induced by the code $\calC_n$ coincides with the output distribution $W_{\rvY^n|\rvS^{n-r}=\mathbf{0}^{n-r}}$ of the constant input $\mathbf{0}^{n-r}$ through $W_{\rvY^n|\rvS^{n-r}}$. Moreover, since $W_{\rvY^n|\rvS^{n-r}}$ is symmetric and $G_n$ is full-rank, the output of the channel $W_{\rvY^n|\rvS^{n-r}}$ to a uniformly distributed input on $\{0,1\}^{n-r}$ has the desired output distribution $q_{{\rvY}^n}$. Hence, applying Lemma~\ref{lm:symmetric_channels} to the channel $W_{\rvY^n|\rvS^{n-r}}$, we find that 
\begin{align*}
\avgD{p_{\rvY^n}}{q_{{\rvY}^n}}&=\avgD{W_{{\rvY^n}|\rvS^{n-r}=\mathbf{0}^{n-r}}}{ q_{{\rvY}^n}}
=C(W_{\rvY^n|\rvS^{n-r}}),
\end{align*}
so that $\lim_{n\rightarrow\infty}\avgD{p_{\rvY^n}}{q_{{\rvY}^n}} = 0$. Pinsker's inequality then ensures that
\begin{align*}
\lim_{n\rightarrow\infty}  \V{p_{\rvY^n},q_{{\rvY}^n}}= 0
\end{align*}
\end{proof}
\begin{remark}
  \label{rk:1}
  The choice of frozen bits set at $\mathbf{0}^{n-r}$ is arbitrary. The choice of a different coset code characterized by $\mathbf{u}_{F}$ in place of $\mathbf{0}^{n-r}$ does not alter the reasoning. In particular, the symmetry of the channel $W_{\rvY^n|\rvS^{n-r}}$ and Lemma~\ref{lm:symmetric_channels} still hold.
\end{remark}
\section{Strong Coordination with Polar Codes}
\label{sec:polar-strong-coord}

\subsection{Strong coordination for a two-node network}
\label{sec:coord-two-node}
The problem of strong coordination for the two-node network~\cite{Cuff2008} is illustrated in Figure~\ref{fig:network}. Node $\rvX$ with actions distributed according to $q_{\rvX^n}\sim\prod_{i=1}^nq_\rvX$ and given by nature wishes to coordinate with node $\rvY$ to obtain the joint distribution of actions $q_{\rvX^n\rvY^n}\sim\prod_{i=1}^nq_{\rvX\rvY}$. Nodes $\rvX$ and $\rvY$ have access to an independent source of common randomness, which provides uniform random numbers in $\intseq{1}{2^{nR_0}}$, and node $\rvX$ transmits messages in $\intseq{1}{2^{nR}}$ to node $\rvY$. Specifically, a $(2^{nR},2^{nR_0},n)$ coordination code $\calC_n$ for this network consists of a stochastic encoding function
\begin{align*}
  f:\calX^n\times\intseq{1}{2^{nR_0}}\rightarrow \intseq{1}{2^{nR}}
\end{align*}
and of a stochastic decoding function
\begin{align*}
  g:\intseq{1}{2^{nR}}\times\intseq{1}{2^{nR_0}}\rightarrow\calY^n.
\end{align*}
\begin{figure}[t]
  \centering
      \scalebox{0.7}{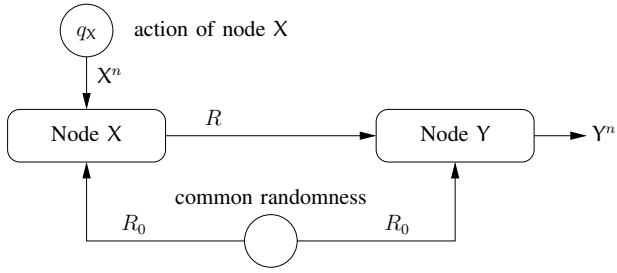}
  \caption{Coordination for two-node network}
  \label{fig:network}
\end{figure} 
We let $\rvU_0\in\intseq{1}{2^{nR_0}}$ denote the common randomness
and ${p}_{\rvX^ng(f(\rvX^n,\rvU_0))}$ be the distribution induced by the coordination code. A coordination $q_{\rvX\rvY}$ is achievable with rates $(R,R_0)$ if there exists a sequence of $(2^{nR},2^{nR_0},n)$ coordination codes $\{\calC_n\}_{n\geq 1}$ such that
\begin{align*}
  \lim_{n\rightarrow\infty}\V{p_{\rvX^ng(f(\rvX^n),\rvU_0)},q_{\rvX^n \rvY^n}}=0
\end{align*}
Let $\rvX$ and $\rvY$ be the random variables with joint distribution $q_{\rvX\rvY}$. It is shown in~\cite{Cuff2008} that the set of achievable rates $(R,R_0)$ is the following.
\begin{theorem}[{\cite[Theorem 3.1]{Cuff2008}}]
\label{th:coordination_capacity}
  The set of achievable rates $(R,R_0)$ for coordination $q_{\rvX\rvY}$ is
  \begin{align*}
    \bigcup_{X\rightarrow\rvV\rightarrow\rvY}\left\{(R,R_0):
      \begin{array}{l}
        R+R_0\geq \avgI{\rvX\rvY;\rvV}\\
        R\geq \avgI{\rvX;\rvV}
      \end{array}
\right\}
  \end{align*}
\end{theorem}
\smallskip{}
In the sequel, we restrict our attention to the case where $\calX=~\{0,1\}$, $q_\rvX\sim~\calB(1/2)$, and the conditional distribution of actions $q_{\rvY|\rvX}$ is symmetric.

\subsection{Coding scheme for strong coordination}
\label{sec:coding-scheme}
In this section, we describe the proposed scheme to achieve strong coordination. Let $\rvX$ and $\rvY$ be the random variables with joint distribution $q_{\rvX\rvY}$, and let $\rvV\in\{0,1\}$ be a binary random variable satisfying the following conditions.
\begin{itemize}
\item[\textbf{C1:}] $\rvX\rightarrow\rvV\rightarrow\rvY$ forms a Markov chain;
\item[\textbf{C2:}] the transition probability $W_{\rvX|\rvV}$ corresponds to a binary symmetric channel;
\item[\textbf{C3:}] the transition probability $W_{\rvY|\rvV}$ is symmetric.
\end{itemize}
By assumption, such a random variable $\rvV$ exists and is distributed according to $\calB(1/2)$.

We first construct polar codes of length $n\eqdef 2^m$ for the channel with transition probabilities $W_{\rvY\rvX|\rvV}$ as follows. 
\begin{itemize}
\item  For the symmetric channel $W_{\rvY\rvX|\rvV}$, and for $i\in\intseq{1}{n}$, we let $\overline{W}_n^{(i)}$ be the corresponding set of bit channels. We define the sets 
\begin{align}
  \calG_{\rvY\rvX|\rvV} &\eqdef\left\{i\in\intseq{1}{n}:C(\overline{W}_n^{(i)})\geq 2^{-n^\beta}\right\},\nonumber{}\\
  \calB_{\rvY\rvX|\rvV}&\eqdef\intseq{1}{n}\setminus \calG_{\rvY\rvX|\rvV}.\label{eq:set_1}
\end{align}
\item For the symmetric channel $W_{\rvX|\rvV}$, and for $i\in\intseq{1}{n}$, we let $\widetilde{W}_n^{(i)}$ be the corresponding set of bit channels. We define the sets 
\begin{align}
  \calG_{\rvX|\rvV} &\eqdef\left\{i\in\intseq{1}{n}:C(\widetilde{W}_n^{(i)})\geq 2^{-n^\beta}\right\},\nonumber{}\\
  \calB_{\rvX|\rvV}&\eqdef\intseq{1}{n}\setminus \calG_{\rvX|\rvV}.\label{eq:set_2}
\end{align}
\end{itemize}
The sets defined in Eq.~\eqref{eq:set_1} and~\eqref{eq:set_2} satisfy the following property.
\begin{lemma}
  \label{lm:degradation}
$\calG_{\rvX|\rvV}\subset \calG_{\rvY\rvX|\rvV}$ and $\calB_{\rvY\rvX|\rvV}\subset\calB_{\rvX|\rvV}$.
\end{lemma}
\begin{proof}
  The channel $W_{\rvX|\rvV}$ is physically degraded with respect to the channel $W_{\rvY\rvX|\rvV}$. Therefore,~\cite[Lemma 21]{Korada2010} guarantees that, for all $i\in\intseq{1}{n}$,  $\widetilde{W}_n^{(i)}$ is degraded with respect to $\overline{W}_n^{(i)}$, so that
  \begin{align*}
    C(\widetilde{W}_n^{(i)})\leq    C(\overline{W}_n^{(i)}). 
    \tag*{\QED}
  \end{align*}
  \let\QED\relax
\end{proof}
Consequently, the sets $F_1$, $F_2$ and $F_3$ defined as
\begin{align*}
  F_1&\eqdef \calB_{\rvY\rvX|\rvV},\\
  F_2&\eqdef \calG_{\rvY\rvX|\rvV}\cap\calB_{\rvX|\rvV},\\
  F_3&\eqdef \calG_{\rvY\rvX|\rvV}\cap\calG_{\rvX|\rvV}.
\end{align*}
form a partition of $\intseq{1}{n}$, which is illustrated in Figure~\ref{fig:nested}.

\begin{figure}[t]
  \centering
      \scalebox{0.8}{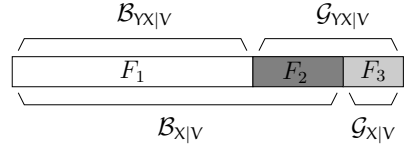}
  \caption{Illustration of partition sets $F_1$, $F_2$, $F_3$ (after reordering of indices).}
  \label{fig:nested}
\end{figure} 
We now exploit these sets to construct a coordination code. The bits in positions $F_1$ are frozen bits with values $\mathbf{u}_{F_1}~=~\mathbf{0}_{F_1}$ fixed at all times. The encoding and decoding procedures are then the following. \smallskip

\noindent \textbf{Operation at node $\rvX$.} To encode a sequence of binary actions $\mathbf{x}\in\calX^n$ provided by nature, node $\rvX$ performs successive-cancellation (SC) encoding to determine the value of the bits $\mathbf{u}_{F_3}$ in $F_3$, using the bits $\mathbf{u}_{F_2}$ from the common randomness in positions $F_2$ and the frozen bits $\mathbf{u}_{F_1}$ in position $F_1$. Specifically, the probability of obtaining a bit $u_i$ during SC encoding is the following~\cite{Korada2010}.
\begin{align} \label{SC_encoding_fixed}
\tilde{p}(u_i|\mathbf{x},\mathbf{u}_1^{i-1})=\begin{cases} 
1 & i \in F_1, u_i=(\mathbf{u}_{F_1})_i \\
0 & i \in F_1, u_i \neq (\mathbf{u}_{F_1})_i \\
\frac{1}{2} & i \in F_2 \\
\frac{L_n^{(i)}(\mathbf{x},\mathbf{u}_1^{i-1})}{1+L_n^{(i)}(\mathbf{x},\mathbf{u}_1^{i-1})} & i \in F_3, u_i=0\\
\frac{1}{1+L_n^{(i)}(\mathbf{x},\mathbf{u}_1^{i-1})} & i \in F_3, u_i=1
\end{cases}
\end{align}  
where
\begin{align*}
  L_n^{(i)}(\mathbf{x},\mathbf{u}_1^{i-1})\eqdef\frac{\widetilde{W}_n^{(i)}(\mathbf{x},\mathbf{u}_1^{i-1}|0)}{\widetilde{W}_n^{(i)}(\mathbf{x},\mathbf{u}_1^{i-1}|1)}.
\end{align*}

The bits in $F_3$ are then transmitted to node $\rvY$. Note that the encoding complexity is that of SC encoding, which is $O(n\log n)$.\smallskip{}

\noindent{}\textbf{Operation at node $\rvY$.} To create a sequence of coordinated actions
  $\mathbf{y}\in\calY^n$, node $\rvY$ creates a vector $\mathbf{u}$
  with frozen bits $\mathbf{u}_{F_1}$, common randomness bits
  $\mathbf{u}_{F_2}$, and received bits $\mathbf{u}_{F_3}$ in
  positions $F_1$, $F_2$, $F_3$, respectively. It then computes the
  vector $\mathbf{u}G_n$, and simulates its transmission over a
  memoryless channel with transition probabilities
  $W_{\rvY|\rvV}$. The resulting vector $\mathbf{y}$ is used as the
  sequence of coordinated actions. The encoding complexity is again $O(n\log n)$.

\begin{remark}
  \label{rk:2}
  Nodes $\rvX$ and $\rvY$ require randomness to perform either SC encoding or simulate a memoryless channel. The evaluation of encoding complexity implicitly assumes that the cost of generating randomness bit-wise is $O(n)$ in both cases.
\end{remark}
The constructed scheme operates at rate $R\eqdef\frac{\card{F_3}}{n}$ between nodes $\rvX$ and $\rvY$ and requires a rate $R_0\eqdef \frac{\card{F_2}}{n}$ of common randomness. Our main result, which we establish in Section~\ref{sec:proof-proposition}, is the following.\smallskip{}
\begin{proposition}
  \label{prop:stonrg_coordination}
    For any random variable $\rvV$ satisfying the conditions \textbf{C1}, \textbf{C2}, and \textbf{C3}, the coordination $q_{\rvX\rvY}$ is achievable with any rates $(R,R_0)$ such that 
  \begin{align*}
    R+R_0>C(W_{\rvY\rvX|\rvV})\quad\text{and}\quad R>C(W_{\rvX|\rvV})
  \end{align*}
\end{proposition}
\smallskip{}

\subsection{Proof of Proposition~\ref{prop:stonrg_coordination}}
\label{sec:proof-proposition}

 The proof is a constructive counterpart of the information-theoretic proof in~\cite{Cuff2008}. We first define the distribution $\tilde{p}$ induced by the encoding/decoding procedures described in Section~\ref{sec:coding-scheme}. By definition,
 \begin{multline*}
   \tilde{p}(\mathbf{u}_{F_2},\mathbf{u}_{F_3},\mathbf{x},\mathbf{y})\\
 \eqdef \frac{1}{2^{\abs{F_2}}}\prod_{i\in F_3}\tilde{p}(u_i|\mathbf{x},\mathbf{u}_1^{i-1})q_{\rvX^n}(\mathbf{x})W_{\rvY^n|\rvV^n}(\mathbf{y}|\mathbf{u}G_n),
 \end{multline*}
 where the vector $\mathbf{u}$ is such that $\mathbf{u}_{F_1}=\mathbf{0}_{F_1}$. We also define the distribution $\hat{p}$ induced by the nested polar code with uniform inputs transmitted over the symmetric channel $W_{\rvY\rvX|\rvV}$ (see Section~\ref{sec:coding-scheme-chann}); we have
 \begin{multline*}
       \hat{p}(\mathbf{u}_{F_2},\mathbf{u}_{F_3},\mathbf{x},\mathbf{y}) \\
       \eqdef \frac{1}{2^{\abs{F_2}}}\frac{1}{2^{\abs{F_3}}}W_{\rvX^n|\rvV^n}(\mathbf{x}|\mathbf{u}G_n) W_{\rvY^n|\rvV^n}(\mathbf{y}|\mathbf{u}G_n).
     \end{multline*}
     By applying the triangle inequality repeatedly, we upper bound the variational distance between the induced distribution $\tilde{p}(\mathbf{x},\mathbf{y})$ and the target coordination $q(\mathbf{x},\mathbf{y})$ as follows.  
\begin{align}
  &\sum_{\mathbf{x},\mathbf{y}}\abs{\tilde p(\mathbf{x},\mathbf{y})-q(\mathbf{x},\mathbf{y})}\nonumber{}\\
  &\leq   \sum_{\mathbf{x},\mathbf{y}}\abs{\tilde p(\mathbf{x},\mathbf{y})-\hat p(\mathbf{x},\mathbf{y})}+  \sum_{\mathbf{x},\mathbf{y}}\abs{\hat p(\mathbf{x},\mathbf{y})-q(\mathbf{x},\mathbf{y})}\nonumber{}\\
  &\leq \sum_{\mathbf{x},\mathbf{y},\mathbf{u}_{F_2},\mathbf{u}_{F_3}}\abs{\tilde p(\mathbf{u}_{F_2},\mathbf{u}_{F_3},\mathbf{x},\mathbf{y})-\hat p(\mathbf{u}_{F_2},\mathbf{u}_{F_3},\mathbf{x},\mathbf{y})}\nonumber{}\\
  &\phantom{--------------}+  \sum_{\mathbf{x},\mathbf{y}}\abs{\hat p(\mathbf{x},\mathbf{y})-q(\mathbf{x},\mathbf{y})}\nonumber{}\\
  &\stackrel{(a)}{=}\sum_{\mathbf{x},\mathbf{u}_{F_2},\mathbf{u}_{F_3}}\abs{\tilde p(\mathbf{u}_{F_2},\mathbf{u}_{F_3},\mathbf{x})-\hat p(\mathbf{u}_{F_2},\mathbf{u}_{F_3},\mathbf{x})}\nonumber{}\\
  &\phantom{--------------}+  \sum_{\mathbf{x},\mathbf{y}}\abs{\hat p(\mathbf{x},\mathbf{y})-q(\mathbf{x},\mathbf{y})}\nonumber{}\\
  &\eqdef \V{\tilde{p}_{\rvU_{F_2}\rvU_{F_3}\rvX^n},\hat{p}_{\rvU_{F_2}\rvU_{F_3}\rvX^n}}+\V{\hat{p}_{\rvX^n\rvY^n},q_{\rvX^n\rvY^n}},\label{eq:bounds_vd}
\end{align}
where equality $(a)$ follows from the definition of $\hat{p}$ and $\tilde{p}$. We first establish that, as $n$ goes to infinity, the coding scheme operates at the sum rate in Proposition~\ref{prop:stonrg_coordination} and $\V{\hat{p}_{\rvX^n\rvY^n},q_{\rvX^n\rvY^n}}$ vanishes. \smallskip{}

\begin{lemma}
  \label{lm:1}
  The sequence of coding schemes satisfies
    \begin{align}
    &\lim_{n\rightarrow\infty}R_0+R = C(W_{\rvY\rvX|\rvV}),\\
\text{and}\quad     &\lim_{n\rightarrow\infty}\V{\hat{p}_{\rvX^n\rvY^n},q_{\rvX^n\rvY^n}}=0.\label{eq:total_resolvability}
  \end{align}
\end{lemma}
\begin{proof}
  By recalling that $F_2\cup F_3 \eqdef \calG_{\rvX\rvY|\rvV}$ and that the bits in positions $F_2$ and $F_3$ are i.i.d $\calB(1/2)$ random bits, Proposition~\ref{prop:resolvablity_polar} guarantees that the coding scheme is a resolvability code for $(W_{\rvY\rvX|\rvV},q_{\rvX})$ with a resolution rate satisfying
$\lim_{n\rightarrow\infty}\frac{1}{n}\abs{\calG_{\rvY\rvX|\rvV}} = C(W_{\rvY\rvX|\rvV})$.
\end{proof}
Next, we show that, as $n$ goes to infinity, the coding scheme achieves the communication rate $R$ of Proposition~\ref{prop:stonrg_coordination} and the \emph{average} over all possible choices of frozen bits $\mathbf{u}_{F_1}$ of $\V{\tilde{p}_{\rvU_{F_2}\rvU_{F_3}\rvX^n},\hat{p}_{\rvU_{F_2}\rvU_{F_3}\rvX^n}}$ vanishes, as well. \smallskip{}

\begin{lemma}
  \label{lm:2}
  The sequence of coding schemes satisfies
  \begin{align}
    \label{eq:1}
        &\lim_{n\rightarrow\infty}R = C(W_{\rvX|\rvV}),\\
    \text{and}\quad &\lim_{n\rightarrow\infty}\E[\rvU_{F_1}]{\V{\tilde{p}_{\rvU_{F_2}\rvU_{F_3}\rvX^n},\hat{p}_{\rvU_{F_2}\rvU_{F_3}\rvX^n}}}=0.
  \end{align}
\end{lemma}
\begin{proof}
  Since $F_3\eqdef \calG_{\rvX|\rvV}$ and since the bits in position $F_3$ are i.i.d. $\calB(1/2)$ random bits, Proposition~\ref{prop:resolvablity_polar} and Remark~\ref{rk:1} ensure that $\lim_{n\rightarrow\infty}\frac{1}{n}{\abs{\calG_{\rvX|\rvV}}} = C(W_{\rvX|\rvV})$.

  We now define two new distributions on $\calU^n\times\calX^n$ as follows.
  \begin{align}
    \label{eq:3}
\hat{P}(\mathbf{u},\mathbf{x})&\eqdef\frac{1}{2^{\card{F_1}}}\frac{1}{2^{\card{F_2}}}\frac{1}{2^{\card{F_3}}}W_{\rvX^n|\rvV^n}(\mathbf{x}|\mathbf{u}G_n),\\
\tilde{P}(\mathbf{u},\mathbf{x})&\eqdef q_{\rvX^n}(\mathbf{x})\prod_{i=1}^n\tilde{P}(u_i|\mathbf{x}\mathbf{u}_{1}^{i-1}),
  \end{align}
  where
  \begin{align*}
    \tilde{P}(u_i|\mathbf{x}\mathbf{u}_{1}^{i-1})\eqdef\left\{
      \begin{array}{l}
        \frac{1}{2}\text{ for }i\in F_1\cup F_2\\
        \tilde{p}(u_i|\mathbf{x}\mathbf{u}_{1}^{i-1})\text{ for }i\in F_3.
      \end{array}
\right.
  \end{align*}
  \begin{remark}
    An important property shown in~\cite{Korada2010} is that, for
    $i\in F_3$, we have
    \begin{align*}
      \tilde{P}(u_i|\mathbf{x}\mathbf{u}_{1}^{i-1}) \eqdef
      \tilde{p}(u_i|\mathbf{x}\mathbf{u}_{1}^{i-1})=\hat{P}(u_i|\mathbf{x}\mathbf{u}_{1}^{i-1}).
    \end{align*}\label{rk:3}
  \end{remark}
One can check that
\begin{align*}
  \E[\rvU_{F_1}]{\V{\tilde{p}_{\rvU_{F_2}\rvU_{F_3}\rvX^n},\hat{p}_{\rvU_{F_2}\rvU_{F_3}\rvX^n}}}=\V{\tilde{P}_{\rvU^n\rvX^n},\hat{P}_{\rvU^n\rvX^n}}.
\end{align*}
We now develop an upper bound for $\V{\tilde{P}_{\rvU^n\rvX^n},\hat{P}_{\rvU^n\rvX^n}}$. Note that
\begin{align}
  &\V{\tilde{P}_{\rvU^n\rvX^n},\hat{P}_{\rvU^n\rvX^n}}\nonumber{}\\
    &=
    \sum_{\mathbf{u},\mathbf{x}}\abs{q(\mathbf{x})\prod_{i=1}^n
      \tilde{P}(u_i|\mathbf{x},\mathbf{u}_1^{i-1})
      -\hat{P}(\mathbf{x})\prod_{i=1}^n
      \hat{P}(u_i|\mathbf{x},\mathbf{u}_1^{i-1})}\nonumber{}\\
    &\leq \underbrace{\sum_{\mathbf{u},\mathbf{x}}\abs{q(\mathbf{x})-\hat{P}(\mathbf{x})}\prod_{i=1}^n
      \tilde{P}(u_i|\mathbf{x},\mathbf{u}_1^{i-1})}_{\eqdef A_n}\nonumber{}\\
      &\phantom{--}+
      \underbrace{\sum_{\mathbf{u},\mathbf{x}}\hat{P}(\mathbf{x})\abs{\prod_{i=1}^n \tilde{P}(u_i|\mathbf{x},\mathbf{u}_1^{i-1})-\prod_{i=1}^n\hat{P}(u_i|\mathbf{x},\mathbf{u}_1^{i-1})}}_{\eqdef B_n}\label{eq:bound_avg_vd}
    \end{align}
    Since $\hat{P}_{\rvU^n}$ is uniform on $\calU^n$ and since $G_n$ defines a bijective map from $\calU^n$ to $\calX^n$, the distribution $\hat{P}_{\rvX^n}$ is also uniform and the term $A_n$ on the right-hand side of Eq.~\eqref{eq:bound_avg_vd} is zero. By applying a telescoping equality to the term $B_n$ as in the proof of \cite[Lemma 4]{Korada2010}, and by recalling that
$\forall i \in F_3$, $\hat{P}(u_i|\mathbf{x},\mathbf{u}_1^{i-1})=\tilde{P}(u_i|\mathbf{x},\mathbf{u}_1^{i-1})$, and that $\forall i \in F_1 \cup F_2$, $\hat{P}(u_i|\mathbf{x},\mathbf{u}_1^{i-1})=\frac{1}{2}$,
we obtain
\begin{align}
   B_n &\leq \sum_{i\in F_1\cup
      F_2}\sum_{u_i,\mathbf{u}_1^{i-1},\mathbf{x}}\hat{P}(\mathbf{x})\hat{P}(\mathbf{u}_1^{i-1}|\mathbf{x})\abs{\frac{1}{2}-\hat{P}(u_i|\mathbf{u}_1^{i-1}\mathbf{x})}\nonumber\\
    &=\sum_{i\in F_1\cup F_2}\sum_{u_i,\mathbf{u}_1^{i-1},\mathbf{x}}\abs{\frac{1}{2}\hat{P}(\mathbf{u}_1^{i-1},\mathbf{x})-\hat{P}(u_i,\mathbf{u}_1^{i-1}\mathbf{x})}\nonumber\\
    &=\frac{1}{2}\sum_{i\in F_1\cup F_2}\sum_{u_i,\mathbf{u}_1^{i-1},\mathbf{x}}\abs{\hat{P}(\mathbf{u}_1^{i-1},\mathbf{x})-\widetilde{W}_{n}^{(i)}(\mathbf{x},\mathbf{u}_1^{i-1}|u_i)}\nonumber\\
    &\eqdef\frac{1}{2}\sum_{i\in F_1\cup F_2}\sum_{u_i}\V{\hat{P}_{\rvX^n\rvU_1^{i-1}},\widetilde{W}_{\rvX^n\rvU_1^{i-1}|\rvU_i=u_i}}\label{eq:bound_avg_vd_2}
  \end{align}
  By noting that
\begin{align*}
\hat{P}(\mathbf{u}_1^{i-1},\mathbf{x}) = \frac{1}{2}\sum_{u\in\{0,1\}}\widetilde{W}_{n}^{(i)}(\mathbf{x},\mathbf{u}_1^{i-1}|u_i=u),
\end{align*}
and since the bit-channels $\widetilde{W}_n^{(i)}$ are symmetric~\cite[Proposition 13]{Arikan2009}, we can argue as in the proof of Proposition~\ref{prop:resolvablity_polar} that for any $u\in\{0,1\}$ and $i\in F_1\cup F_2$,
\begin{align*}
  \avgD{\widetilde{W}_{\rvX^n\rvU_1^{i-1}|\rvU_i=u}}{\hat{P}_{\rvX^n\rvU_1^{i-1}}}=C(\widetilde{W}_n^{(i)})\leq 2^{-n^{\beta}}.
\end{align*}
Using Pinsker's inequality, we obtain
\begin{align*}
  \V{\hat{P}_{\rvX^n\rvU_1^{i-1}},\widetilde{W}_{\rvX^n\rvU_1^{i-1}|\rvU_i=u}}\leq 2^{-\frac{1}{2}n^{\beta}} \sqrt{2\ln 2},
\end{align*}
and we conclude that $B_n\leq n2^{-\frac{1}{2}n^{\beta}} \sqrt{2\ln 2}$. Therefore,
\begin{align*}
  \lim_{n\rightarrow\infty}   \V{\tilde{P}_{\rvU^n\rvX^n},\hat{P}_{\rvU^n\rvX^n}}\leq \lim_{n\rightarrow\infty}B_n = 0.
\end{align*}
\end{proof}
Finally we show that $\V{\tilde{p}_{\rvU_{F_2}\rvU_{F_3}\rvX^n\rvY^n},\hat{p}_{\rvU_{F_2}\rvU_{F_3}\rvX^n\rvY^n}}$ is independent of the value of the frozen bits $\mathbf{u}_{F_1}$.\smallskip{}

\begin{lemma}
  \label{lm:3}
  \begin{multline}
    \label{eq:2}
    \E[\rvU_{F_1}]{\V{\tilde{p}_{\rvU_{F_2}\rvU_{F_3}\rvX^n\rvY^n},\hat{p}_{\rvU_{F_2}\rvU_{F_3}\rvX^n\rvY^n}}}\\
    = \V{\tilde{p}_{\rvU_{F_2}\rvU_{F_3}\rvX^n\rvY^n},\hat{p}_{\rvU_{F_2}\rvU_{F_3}\rvX^n\rvY^n}}.
  \end{multline}
\end{lemma}
\begin{proof}
Since the channel $W_{\rvX|\rvV}$ is symmetric, there exists a permutation $\pi_1:\calX\rightarrow \calX$ such that $\pi_1=\pi_1^{-1}$ and 
\begin{equation} \label{symmetry}
\forall x \in \mathcal{X}, \quad W_{\rvX|\rvV}(x|0)=W_{\rvX|\rvV}(\pi_1(x)|1).
\end{equation}  
Defining the identity $\pi_0: \mathcal{X} \to \mathcal{X}$ allows us to define an action $\{0,1\}\times\mathcal{X}\to \mathcal{X}$ given by
$$v \cdot x  = \pi_{v}(x).$$
This can be extended component-wise to an action $\{0,1\}^n\times \mathcal{X}^n \to \mathcal{X}^n$ as
$$(v_1,\ldots,v_n) \cdot (x_1,\ldots,x_n)   = (\pi_{v_1}(x_1),\ldots,\pi_{v_n}(x_n)).$$
Therefore, we have $\forall \mathbf{v},\mathbf{w} \in \{0,1\}^n$, $\forall \mathbf{x} \in \mathcal{X}^n$,
\begin{equation} \label{action}
W_{\rvX^n|\rvV^n}(\mathbf{x}|\mathbf{v})=W_{\rvX^n|\rvV^n}(\mathbf{w} \cdot \mathbf{x} |\mathbf{v} \oplus \mathbf{w}) 
\end{equation}
Lemma 8 in~\cite{Korada2010} shows that  $\forall i\in\intseq{1}{n}$
\begin{multline} \label{gauge_transformation}
L_n^{(i)}(\mathbf{w} \cdot \mathbf{x},(\mathbf{w}G_n^{-1})_1^{i-1} \oplus \mathbf{u}_1^{i-1})\\
=\begin{cases}
L_n^{(i)}(\mathbf{x},\mathbf{u}_1^{i-1}) & \text { if }  (\mathbf{w}G_n^{-1})_i=0\\
(L_n^{(i)}(\mathbf{x},\mathbf{u}_1^{i-1}))^{-1} & \text { if }  (\mathbf{w}G_n^{-1})_i=1
\end{cases}
\end{multline}
Now, consider the two SC encodings corresponding to two values of the frozen bits $\breve{\mathbf{u}}_{F_1}$ and $\db{\mathbf{u}}_{F_1}$.
\begin{align} \label{SC_encoder_1}
&\breve{p}(u_i|\mathbf{x},\mathbf{u}_1^{i-1})=\begin{cases} 
1 & i \in F_1, u_i=(\breve{\mathbf{u}}_{F_1})_i \\
0 & i \in F_1, u_i \neq (\breve{\mathbf{u}}_{F_1})_i \\
\frac{1}{2} & i \in F_2 \\
\frac{L_n^{(i)}(\mathbf{x},\mathbf{u}_1^{i-1})}{1+L_n^{(i)}(\mathbf{x},\mathbf{u}_1^{i-1})} & i \in F_3, u_i=0\\
\frac{1}{1+L_n^{(i)}(\mathbf{x},\mathbf{u}_1^{i-1})} & i \in F_3, u_i=1
\end{cases}
\\
\label{SC_encoder_2}
&\db{p}(u_i|\mathbf{x},\mathbf{u}_1^{i-1})=\begin{cases} 
1 & i \in F_1, u_i=(\db{\mathbf{u}}_{F_1})_i \\
0 & i \in F_1, u_i \neq (\db{\mathbf{u}}_{F_1})_i \\
\frac{1}{2} & i \in F_2 \\
\frac{L_n^{(i)}(\mathbf{x},\mathbf{u}_1^{i-1})}{1+L_n^{(i)}(\mathbf{x},\mathbf{u}_1^{i-1})} & i \in F_3, u_i=0\\
\frac{1)}{1+L_n^{(i)}(\mathbf{x},\mathbf{u}_1^{i-1})} & i \in F_3, u_i=1
\end{cases}
\end{align}
Using (\ref{gauge_transformation}) and following the proof of Lemma 9 in~\cite{Korada2010}, one can show by induction that if $\mathbf{w} \in \{0,1\}^n$ is such that 
\begin{equation} \label{compatibility_condition}
\db{\mathbf{u}}_{F_1} \oplus \breve{\mathbf{u}}_{F_1}=(\mathbf{w}G_n^{-1})_{F_1},
\end{equation}
 then $\forall i\in\intseq{1}{n}$, $\forall u_i \in \{0,1\}$, we have
\begin{multline} \label{Lemma313}
\breve{p}(u_i |\mathbf{x},\mathbf{u}_1^{i-1})\\
=\db{p}(u_i \oplus (\mathbf{w}G_n^{-1})_i |\mathbf{w} \cdot \mathbf{x}, (\mathbf{w}G_n^{-1})_1^{i-1} \oplus\mathbf{u}_1^{i-1}).
\end{multline}
Note that, given $\breve{\mathbf{u}}_{F_1}$ and $\db{\mathbf{u}}_{F_1}$, a $\mathbf{w}$ satisfying (\ref{compatibility_condition}) always exists since $G_n$ is one-to-one.

Similarly to Lemma 10 in~\cite{Korada2010}, where it is shown that the average distortion is independent of the choice of frozen bits, we prove that the variational distance is independent of the choice of the frozen bits $\mathbf{u}_{F_1}$. Consider two resolvability codes for the channel $W_{\rvX|\rvV}$ obtained by transmitting i.i.d. $\mathcal{B}(1/2)$ random bits on $F_2$ and $F_3$ and by freezing the bits in $F_1$ to $\db{\mathbf{u}}_{F_1}$ and $\breve{\mathbf{u}}_{F_1}$, respectively. Denote the induced distribution by $\hat{p}_{{\rvX}^n{\rvU}^n|\rvU_{F_1}=\db{\mathbf{u}}_{F_1}}$ and $\hat{p}_{{\rvX}^n{\rvU}^n|\rvU_{F_1}=\breve{\mathbf{u}}_{F_1}}$, respectively. Our goal is to show that
\begin{align*}
\V{\db{p}_{{\rvX}^n{\rvU}^n},\hat{p}_{{\rvX}^n{\rvU}^n|\rvU_{F_1}=\db{\mathbf{u}}_{F_1}}}=\V{\breve{p}_{{\rvX}^n{\rvU}^n},\hat{p}_{{\rvX}^n{\rvU}^n|\rvU_{F_1}=\breve{\mathbf{u}}_{F_1}}}.
\end{align*}
In fact, we have
\begin{multline}
\V{\db{p}_{{\rvX}^n{\rvU}^n},\hat{p}_{{\rvX}^n{\rvU}^n|\rvU_{F_1}=\db{\mathbf{u}}_{F_1}}}\\
\begin{split}
  &=\sum_{\mathbf{x},\mathbf{u}} \frac{1}{2^{\abs{F_2}}}  \mathds{1}_{\left\{\mathbf{u}_{F_1}=\db{\mathbf{u}}_{F_1}\right\}}  \left\vert q(\mathbf{x}) \prod_{i \in F_3}
    \db{p}(u_i|\mathbf{x},\mathbf{u}_1^{i-1})\right.\\
  &\phantom{-------------}\left.-\frac{1}{2^{\abs{F_3}}}W_{\rvX^n|\mathsf{V}^n}(\mathbf{x}|\mathbf{u}G_n)\right\vert.
\end{split}\label{eq:vd_u1}
\end{multline}
Consider the change of variables $\mathbf{u}=\mathbf{v} \oplus \mathbf{w}G_n^{-1}$ and $\mathbf{x} = \mathbf{w} \cdot \mathbf{z}$, where $\mathbf{w}$ satisfies (\ref{compatibility_condition}). Equation~\eqref{eq:vd_u1} becomes
\begin{multline*}
\sum_{\mathbf{z},\mathbf{v}} \frac{1}{2^{\abs{F_2}}} \mathds{1}_{\left\{\mathbf{v}_{F_1}=\breve{\mathbf{u}}_{F_1}\right\}} \\
\left| q(\mathbf{w} \cdot \mathbf{z}) \prod_{i \in F_3} \db{p}(v_i \oplus (\mathbf{w}G_n^{-1})_i|\mathbf{w} \cdot \mathbf{z},\mathbf{v}_1^{i-1} \oplus (\mathbf{w}G_n^{-1})_1^{i-1})
 \right.\\
\left.-\frac{1}{2^{\abs{F_3}}}W_{\rvX^n|\rvV^n}(\mathbf{w}\cdot\mathbf{z}|\mathbf{v}G_n \oplus \mathbf{w})\right|
\end{multline*}
Using Eq.~(\ref{Lemma313}) and Eq.~(\ref{action}), this further simplifies as 
\begin{align*}
&\sum_{\mathbf{z},\mathbf{v}} \frac{1}{2^{\abs{F_2}}} \mathds{1}_{\left\{\mathbf{v}_{F_1}=\breve{\mathbf{u}}_{F_1}\right\}}
\left| q(\mathbf{w} \cdot \mathbf{z}) \prod_{i \in F_3} \breve{p}(v_i |\mathbf{z},\mathbf{v}_1^{i-1}) \right.\\
&\phantom{-------------}\left.-\frac{1}{2^{\abs{F_3}}}W_{\rvX^n|\rvV^n}(\mathbf{z}|\mathbf{v}G_n)\right|\\
&=\V{\breve{p}_{{\rvX}^n{\rvU}^n},\hat{p}_{{\rvX}^n{\rvU}^n|\rvU_{F_1}=\breve{\mathbf{u}}_{F_1}}},
\end{align*} 
where the last inequality follows because $q_{{\rvX}^n}$ is the uniform distribution on $\calX^n$.
\end{proof}
Combining the results of Lemma~\ref{lm:1}, Lemma~\ref{lm:2}, and Lemma~\ref{lm:3} with Eq.~\eqref{eq:bounds_vd}, we conclude that the proposed coding scheme is a resolvability code.

\section{Discussion}
\label{sec:discussion}
In general, the achievable coordination region with polar codes given in Proposition~\ref{prop:stonrg_coordination} is strictly smaller than the coordination capacity region given in Theorem~\ref{th:coordination_capacity} because of the constraints  \textbf{C1}, \textbf{C2}, and \textbf{C3}, on the random variable $\rvV$ (see Section~\ref{sec:coding-scheme}). 

As a first illustration, consider the situation in which $\calX=\{0,1\}$, $\calY=\{0,?,1\}$, $q_{\rvX}\sim\calB(\frac{1}{2})$ and $q_{\rvY|\rvX}$ corresponds to the concatenation of a binary symmetric channel with cross-over probability $p$ with a binary erasure channel with erasure probability $\epsilon$. In other words, the transition probability matrix corresponding to $q_{\rvY|\rvX}$ is
\begin{align*}
  \left(
  \begin{array}{ccc}
    (1-p)(1-\epsilon)&\epsilon&p(1-\epsilon)\\
    p(1-\epsilon)&\epsilon&(1-p)(1-\epsilon)
  \end{array}
\right)
\end{align*}
Because of condition \textbf{C2} (see Section~\ref{sec:coding-scheme}), one can show that the boundary of the region of achievable rates $(R,R_0)$ is characterized by
\begin{align*}
  R &\geq 1-\Hb{q}\\
R_0+R &\geq (1-\epsilon)\left(\Hb{p}-\Hb{\frac{p-q}{1-2q}}\right) +1-\Hb{q}
\end{align*}
for $q\in[0,\min(\frac{1}{2},p)]$. On the other hand, it is not difficult to show that the following rates are also admissible by Theorem~\ref{th:coordination_capacity}.
\begin{align*}
  R&\geq (1-\nu)(1-\Hb{p})\\
  R_0 +R &\geq \Hb{\epsilon}+(1-\epsilon)\Hb{p} +(1-\nu)\Hb{\frac{\epsilon-\nu}{1-\nu}}\\
  &\phantom{---------}+(1-\nu)(1-\Hb{p})
\end{align*}
for $\nu\in[0,\min(1,\epsilon)]$. The regions achievable with and without polar codes are illustrated in Figure~\ref{fig:example}, for the case $\epsilon=0.4$ and $p=0.15$.

\begin{figure}[t]
  \centering
  \includegraphics[width=0.9\linewidth]{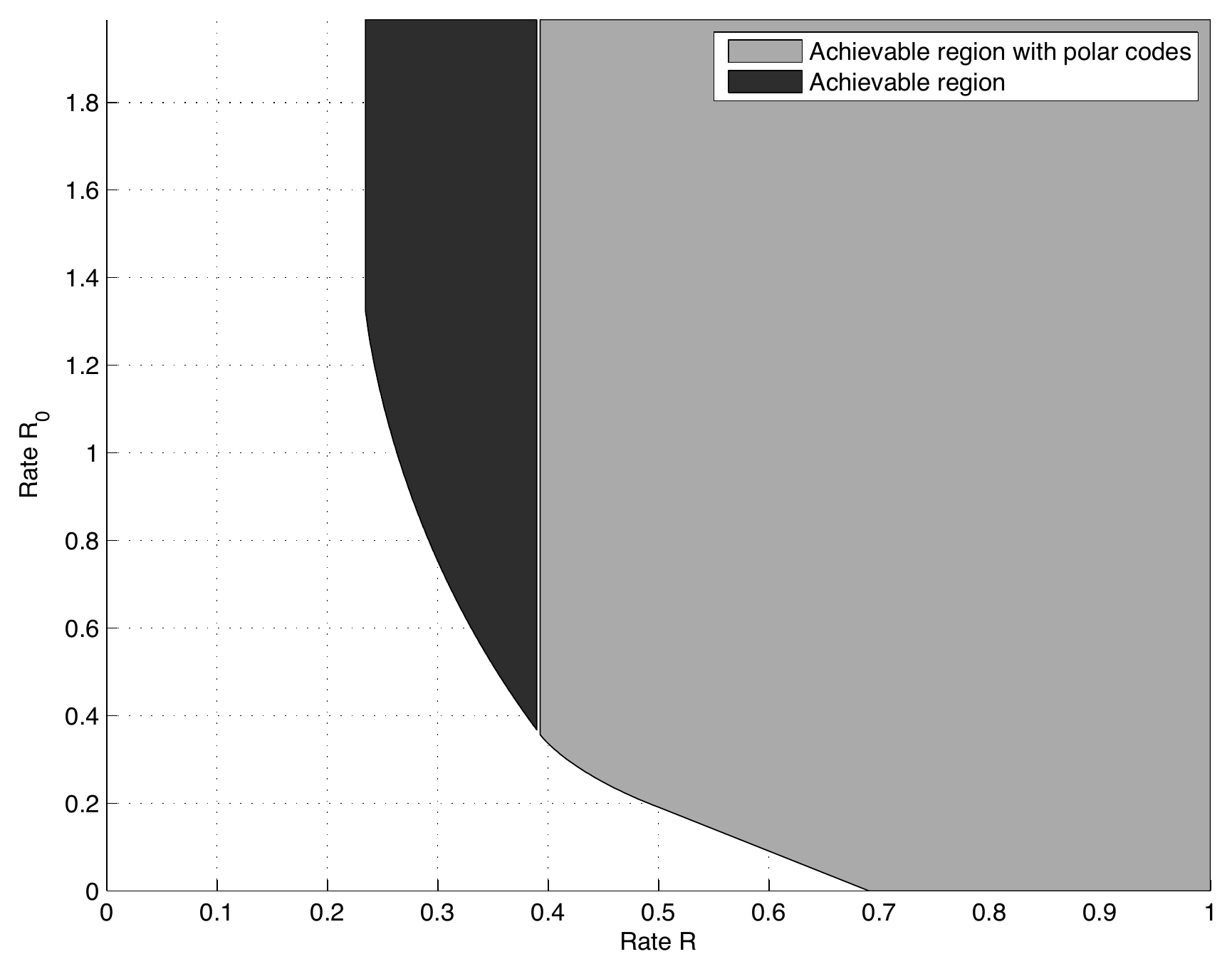}
  \caption{Example of achievable rates for coordination with and without polar codes.}
  \label{fig:example}
\end{figure}

As a second illustration, consider the situation in which $\calX=\{0,1\}$, $\calY=\{0,?,1\}$, and $q_{\rvY|\rvX}$ corresponds to binary erasure channel with erasure probability $\epsilon$. The coordination capacity region for this model is characterized in~\cite{Cuff2008}, and it is shown that the optimal choice of $\rvV$ such that $\rvX\rightarrow\rvV\rightarrow\rvY$ forms a Markov chain is a ternary random variable. In contrast, one can check that the only possible choice of such a $V$ satisfying the constraints \textbf{C1}, \textbf{C2}, and \textbf{C3} is $\rvV=\rvX$. Consequently, the achievable coordination rate with polar codes is the trivial region $\{(R,R_0):R_0\geq 0, R\geq 1$, which is achievable without any coding. The regions are illustrated in Figure~\ref{fig:second_example}.

\begin{figure}[t]
  \centering
  \includegraphics[width=0.9\linewidth]{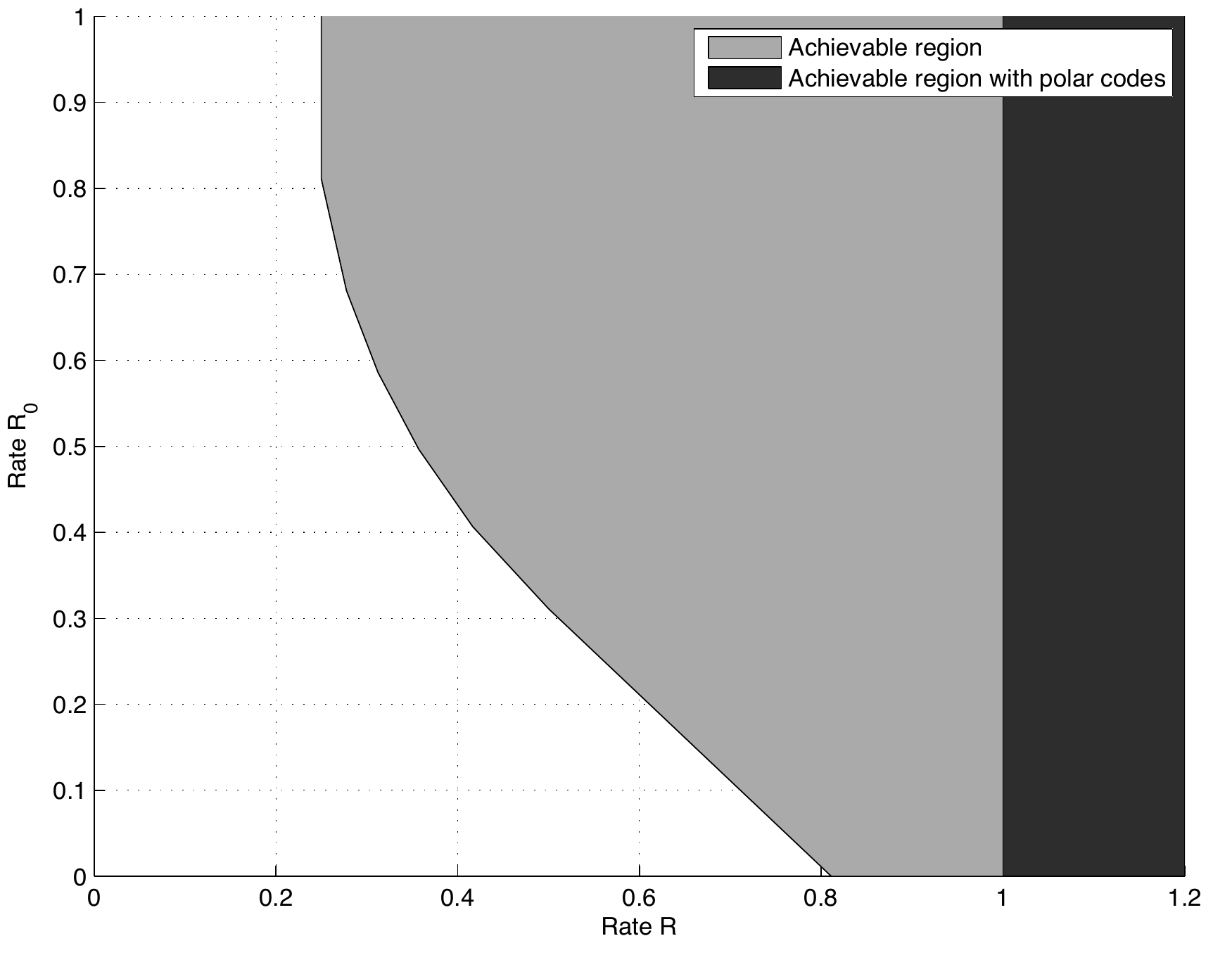}
  \caption{Example of achievable rates for coordination with and without polar codes.}
  \label{fig:second_example}
\end{figure}

The generalization of the results beyond binary actions at node $\rvX$ can be carried out by leveraging known results about non-binary polar codes. However, the generalization to non-uniform actions and asymmetric channels seems much more challenging, since the proofs used in this paper heavily rely on the symmetry properties and uniformity of the actions to coordinate. Finding an explicit coordination scheme in a more general case remains an open problem and will be the topic of future research. 

\section*{Acknowledgement}
\label{sec:acknowledgement}
The research of M. Bloch was supported in part by a PEPS grant from the Centre National de la Recherche Scientifique. The research of L. Luzzi was supported by the European Union Seventh Framework Program under grant agreement PIEF-GA-2010-274765. The research of J. Kliewer was supported by the U.S.~National Science Foundation under grants CCF-0830666 and CCF-1017632.
\newpage
\IEEEtriggeratref{0}
\bibliographystyle{IEEEtran}
\bibliography{itw2012}

\begin{thebibliography}{10}
\providecommand{\url}[1]{#1}
\csname url@samestyle\endcsname
\providecommand{\newblock}{\relax}
\providecommand{\bibinfo}[2]{#2}
\providecommand{\BIBentrySTDinterwordspacing}{\spaceskip=0pt\relax}
\providecommand{\BIBentryALTinterwordstretchfactor}{4}
\providecommand{\BIBentryALTinterwordspacing}{\spaceskip=\fontdimen2\font plus
\BIBentryALTinterwordstretchfactor\fontdimen3\font minus
  \fontdimen4\font\relax}
\providecommand{\BIBforeignlanguage}[2]{{%
\expandafter\ifx\csname l@#1\endcsname\relax
\typeout{** WARNING: IEEEtran.bst: No hyphenation pattern has been}%
\typeout{** loaded for the language `#1'. Using the pattern for}%
\typeout{** the default language instead.}%
\else
\language=\csname l@#1\endcsname
\fi
#2}}
\providecommand{\BIBdecl}{\relax}
\BIBdecl

\bibitem{Cuff2010}
P.~W. Cuff, H.~H. Permuter, and T.~M. Cover, ``{C}oordination {C}apacity,''
  \emph{{IEEE} {T}ransactions on {I}nformation {T}heory}, vol.~56, no.~9, pp.
  4181--4206, 2010.

\bibitem{Arikan2009}
E.~Arikan, ``{C}hannel {P}olarization: {A} {M}ethod for {C}onstructing
  {C}apacity-{A}chieving {C}odes for {S}ymmetric {B}inary-{I}nput {M}emoryless
  {C}hannels,'' \emph{{IEEE} {T}ransactions on {I}nformation {T}heory},
  vol.~55, no.~7, pp. 3051--3073, 2009.

\bibitem{Blasco-Serrano2012}
R.~Blasco-Serrano, R.~Thobaben, and M.~Skoglund, ``{P}olar {C}odes for
  {C}oordination in {C}ascade {N}etworks,'' in \emph{Proc. of International
  Zurich Seminar on Communications}, Zurich, Switzerland, March 2012, pp.
  55--58.

\bibitem{Korada2010}
S.~B. Korada and R.~L. Urbanke, ``{P}olar {C}odes are {O}ptimal for {L}ossy
  {S}ource {C}oding,'' \emph{{IEEE} {T}ransactions on {I}nformation {T}heory},
  vol.~56, no.~4, pp. 1751--1768, April 2010.

\bibitem{Han1993}
T.~Han and S.~Verd\'u, ``{A}pproximation {T}heory of {O}utput {S}tatistics,''
  \emph{{IEEE} {T}rans. {I}nf. {T}heory}, vol.~39, no.~3, pp. 752--772, May
  1993.

\bibitem{Wyner1975a}
A.~Wyner, ``{T}he {C}ommon {I}nformation of {T}wo {D}ependent {R}andom
  {V}ariables,'' \emph{{IEEE} {T}rans. {I}nf. {T}heory}, vol.~21, no.~2, pp.
  163--179, March 1975.

\bibitem{Hayashi2006}
M.~Hayashi, ``{G}eneral {N}onasymptotic and {A}symptotic {F}ormulas in
  {C}hannel {R}esolvability and {I}dentification {C}apacity and their
  {A}pplication to the {W}iretap {C}hannels,'' \emph{{IEEE} {T}rans. {I}nf.
  {T}heory}, vol.~52, no.~4, pp. 1562--1575, April 2006.

\bibitem{Bloch2011a}
M.~R. Bloch, ``{A}chieving {S}ecrecy: {C}apacity vs. {R}esolvability,'' in
  \emph{Proc. of IEEE International Symposium on Information Theory}, Saint
  Petersburg, Russia, August 2011, pp. 632--636.

\bibitem{Mahdavifar2011}
H.~Mahdavifar and A.~Vardy, ``{A}chieving the {S}ecrecy {C}apacity of {W}iretap
  {C}hannels {U}sing {P}olar {C}odes,'' \emph{{IEEE}. {T}rans. {I}nf.
  {T}heory}, vol.~57, no.~10, pp. 6428--6443, 2011.

\bibitem{InformationTheoryReliableCommunication}
R.~G. Gallager, \emph{{I}nformation {T}heory and {R}eliable
  {C}ommunication}.\hskip 1em plus 0.5em minus 0.4em\relax Wyley, 1968.

\bibitem{Cuff2008}
P.~Cuff, ``{C}ommunication requirements for generating correlated random
  variables,'' in \emph{Proc. IEEE International Symposium on Information
  Theory ISIT 2008}, 2008, pp. 1393--1397.

\end{thebibliography}

\end{document}